\documentclass[11pt,letterpaper]{article}

\usepackage[USenglish]{babel}
\usepackage[Tol,keep-defaults,no-tikz]{colorblind}
\usepackage{mathtools}
\usepackage{styles/arxiv/custom-arxiv}
\usepackage{microtype}
\usepackage{graphicx}
\usepackage{tikz}
\usepackage{standalone}
\usepackage{doi}

\Crefname{claim}{Claim}{Claims}
\NewEnvironmentCopy{bfsclaim}{claim}
\AddToHook{env/bfsclaim/begin}{\setlength{\emergencystretch}{4em}}
\newcommand{\keywords}[1]{%
  \par\medskip\noindent{\def\and{\textperiodcentered\space}%
  \textbf{Keywords}\enspace #1}\par
}

\NewCommandCopy{\standardoperatorname}{\operatorname}
\RenewDocumentCommand{\operatorname}{s m}{%
  \IfBooleanTF{#1}
    {\standardoperatorname*{\mathsf{#2}}}
    {\standardoperatorname{\mathsf{#2}}}%
}

\usetikzlibrary{
  arrows.meta,
  backgrounds,
  calc,
  decorations.pathmorphing,
  fit,
  positioning,
  shapes.geometric
}

\newcommand{\Nat}{\ensuremath{\mathbb{N}}}

\newcommand{\plgn}{\ensuremath{\mathrm{poly}(\log n)}}
\newcommand{\plglgn}{\ensuremath{\mathrm{poly}(\log\log n)}}

\newcommand{\eps}{\varepsilon}

\title{Breadth-First Search in Succinct Planar Graphs}
\newif\ifanonymous
\anonymousfalse
\ifanonymous
  \author{Anonymous Author(s)}
\else
  \author{Johannes Meintrup\thanks{Funded by the Deutsche Forschungsgemeinschaft
    (DFG, German Research Foundation) -- 571642628.}}
  \affil{THM, University of Applied Sciences Mittelhessen, Gie{\ss}en, Germany}
  \hypersetup{pdfauthor={Johannes Meintrup}}
\fi
\date{}
\hypersetup{
  pdftitle={Breadth-First Search in Succinct Planar Graphs},
  pdfkeywords={Breadth-First Search, Succinct Graph Encodings,
    Space-Efficient Algorithms, Planar Graphs}
}

\begin{document}

\setlength{\unitlength}{1mm}
\newcommand{\Newdot}{{%
  \leavevmode
  \put(0,.63){\circle*{1.5}}%
}}

\maketitle

\begin{abstract}
An encoding for a graph class $\mathcal G$ is called succinct if it uses
$Z_{\mathcal G}(n)+o(Z_{\mathcal G}(n))$ bits, where $Z_{\mathcal G}(n)$ is
the information-theoretic minimum number of bits needed to distinguish
between $n$-vertex graphs of $\mathcal G$; it is called compact if it uses
$O(Z_{\mathcal G}(n))$ bits. For planar graphs,
$Z_{\mathcal G}(n)=\Theta(n)$.

We present a succinct encoding of planar graphs that supports executing a
breadth-first search directly on the encoding. The succinct encoding can be
constructed in expected $O(n)$ time using $O(n)$ bits during construction;
a compact variant can be constructed in deterministic $O(n)$ time using
$O(n)$ bits. Once the encoding is constructed, a BFS from any start vertex
can be computed in $O(n)$ time using $o(n)$ additional bits, including the
space needed to represent the BFS tree. The resulting BFS tree $T$ remains
available for standard tree operations, such as traversal, parent and child
queries, layer queries, and lowest common ancestor queries, in constant time
per query and amortized constant time per output element. The encoding also
supports standard graph queries. Finally, all BFS and tree-query results
generalize to separable graph classes, with graph-class-specific
construction time and space for the encoding.

\end{abstract}

\section{Introduction}
\label{sec:introduction}

Breadth-first search (BFS) is, alongside depth-first search (DFS), one of the
most fundamental graph traversal algorithms. Throughout the paper, we use $n$
to denote the number of vertices of a graph and $m$ to denote the number of
edges.
A standard BFS running in $O(n+m)$
time stores a queue of vertices and a visitation status for every vertex, and
therefore uses $\Theta(n\log n)$ bits in the word-RAM model.
Several works study space-efficient BFS and DFS algorithms that use (almost)
linear space in bits while retaining (almost) linear running
time~\cite{AsanoIKKOOSTU14,BanerjeeCRS18,ElmasryHK15,Hagerup20}. Related work
on reachability pushes the space below linear, at the cost of large polynomial
running time~\cite{IzumiO20,Reingold08}. The space-efficient BFS and DFS
algorithms above primarily output the vertices in visitation order; they do
not keep the resulting BFS or DFS tree available as a data structure.

This missing structural access is a serious limitation when a traversal is used
as a subroutine in more complex applications. Standard graph algorithms often
rely not only on the order in which vertices are visited, but also on direct
access to the resulting tree: parent and child queries, depths or layers,
ancestor queries, and more. In the space-efficient setting, such information
cannot simply be stored for all vertices. Consequently, space-efficient
algorithms using BFS or DFS as subroutines often require substantial
modifications to work around this lack of queries. For example, space-efficient
algorithms for biconnected components~\cite{ChakrabortyRS16,KammerKL19} require
non-trivial modifications of the classical DFS-based algorithm of
Tarjan~\cite{HopcroftT73} which normally has direct access to DFS numbers,
depths, and so-called lowpoints. A recent result of Elberfeld et
al.~\cite{ElberfeldKM25} addresses exactly this issue for DFS on planar and,
more generally, separable graphs. Their work builds on the succinct encoding of
Blelloch and Farzan~\cite{BlellochF10}, which, roughly speaking, decomposes the
graph into very small subgraphs represented by indices into lookup tables;
these subgraphs overlap only in a few so-called boundary vertices. By augmenting
the lookup table with states of a DFS computation, and by storing data
explicitly for the boundary vertices, Elberfeld et al. compute a DFS directly
on the encoding using only $o(n)$ additional bits. The resulting DFS tree
remains available for standard tree queries.

We provide the analogous structure for BFS. After a modified version of the
succinct encoding of Blelloch and Farzan~\cite{BlellochF10} has been
constructed, a BFS from an arbitrary start vertex can be executed in $O(n)$
time using only $o(n)$ additional bits, and the resulting BFS tree remains
available for direct access through common queries. The main technical
difference is that DFS is highly local with respect to the tiny subgraphs used
by the encoding. Once a DFS enters such a subgraph, a table lookup can advance
the DFS through the subgraph in $O(1)$ time. This locality is what allows the
DFS algorithm of Elberfeld et al.~\cite{ElberfeldKM25} to run in time
$O(n/\log^c\log n)$, for some constant $c$, after $O(n)$ preprocessing time.
For BFS, the situation is different: a BFS tree has to preserve shortest-path
distances from the root. Thus, the search cannot simply exhaust a tiny subgraph
locally, since a later entry through another boundary vertex may give shorter
paths to some of its vertices. Nevertheless, we show that the BFS tree can be
computed in $O(n)$ time and kept available inside the encoding using only
$o(n)$ additional space.
Our result fits into the broader line of succinct graph encodings. For a recent
survey, see~\cite{ChakrabortyS2025}. Succinct representations of general
labeled graphs are considered in~\cite{BarbayAHM12}. For planar graphs
and triangulations, see, e.g.,~\cite{AleardiCDS06,ChiangLL01,FuentesNS23,FuentesSV21,GagieHD23}.
Further results concern interval
graphs~\cite{Acan2021,HeMNYWW20}, succinct tree
representations~\cite{RichardRR06,MunroNBW21}, and fundamental succinct data
structures used by such encodings~\cite{RamanRRS07}.

\smallpar{Our Results.}
All graphs considered are connected, unlabeled, simple and undirected.
More on \emph{unlabeled} later.
Our main result is a succinct encoding for planar graphs that supports BFS
directly on the encoding while keeping the computed BFS tree available for
queries. 
\newpage
\begin{theorem}\label{thm:main}
Let $G=(V,E)$ be a connected planar graph. There is a succinct encoding of $G$
that provides the following.
For a given vertex $s\in V$, we can compute
and encode a BFS tree $T$ rooted at $s$ in $O(n)$ time using $o(n)$ additional
bits. 
After the BFS has been computed, the following operations on $T$ are supported:
\begin{itemize}
\item traverse $T$ in preorder, postorder, or BFS order;
\item output the BFS layer $\ell(u)$ of a vertex $u$;
\item output the parent of a vertex $u$ in $T$;
\item iterate over the children of a vertex $u$ in $T$;
\item test whether an edge is a tree edge or a non-tree edge;
\item compute the lowest common ancestor in $T$ of two vertices.
\end{itemize}
The encoding additionally supports adjacency, degree, and neighborhood queries.
All operations run in $O(1)$ time (amortized per output element for
enumerations). Moreover, the encoding can be reset to its initial state in
$o(n)$ time.
\end{theorem}
After construction, BFS and all queries operate solely on the encoding,
without accessing the original input.
As in the succinct encoding of Blelloch and Farzan~\cite{BlellochF10} and its
DFS augmentation by Elberfeld et al.~\cite{ElberfeldKM25}, we distinguish
between succinct and compact variants. The distinction comes from the use of
compressed indexable dictionaries~\cite{RamanRRS07}: replacing them by simpler
compact dictionaries gives deterministic linear construction time, but only a
compact encoding.
\begin{corollary}\label{cor:construction}
    The succinct encoding of \Cref{thm:main} can be constructed in expected
    $O(n)$ time using $O(n)$ bits during construction. A compact variant can be
    constructed in deterministic $O(n)$ time using $O(n)$ bits during construction.
\end{corollary}
All BFS and tree-query results extend to separable graph classes, with the
same running-time and additional-space bounds. The time and space needed to
construct the encoding depend on the graph class; the construction bounds
of \Cref{cor:construction} are specific to planar graphs.
See \Cref{app:general} for the generalization.

\smallpar{A note on unlabeled graphs.}
As is common for succinct graph encodings, our encoding is for unlabeled
graphs: input labels are not preserved. Preserving vertex labels
requires $\Theta(n\log n)$ bits, whereas unlabeled planar graphs, and
many similar graph classes, admit $\Theta(n)$-bit encodings. Thus, the encoding
stores a graph isomorphic to the input graph.

\smallpar{Outline.}
\Cref{sec:prelim} gives the necessary preliminaries. In particular, it recalls
the succinct encoding of Blelloch and Farzan~\cite{BlellochF10} and extensions
thereof by Elberfeld et al.~\cite{ElberfeldKM25}.
\Cref{sec:bfs} describes how to execute a BFS directly on the succinct encoding
in $O(n)$ time using $o(n)$ additional bits, and how to keep the resulting BFS
tree available for standard tree queries. 
Information on the construction
time and space of our encoding (i.e., \Cref{cor:construction}) can be found in
Section~\ref{app:construction}. The generalization to other graph classes
is described in Section~\ref{app:general}.
Missing proofs
can be found in the appendix.

\section{Preliminaries} \label{sec:prelim}
We work in the standard word-RAM model with word size $\Omega(\log n)$. We
assume the input is in read-only memory and measure space in the read/write
working memory. We use $[x]\coloneqq\{1,\ldots,x\}$ for $x\in\Nat$. For ease of
presentation, we assume that all input graphs are connected. For disconnected
graphs, connected components can be identified during construction, and our
results applied to each component individually. 
An encoding for a graph class $\mathcal{G}$ is called succinct if it uses
$Z_{\mathcal G}(n)+o(Z_{\mathcal G}(n))$ bits, where $Z_{\mathcal G}(n)$ is the
information-theoretic minimum number of bits needed to distinguish between
$n$-vertex graphs of $\mathcal G$; it is called compact if it uses
$O(Z_{\mathcal G}(n))$ bits. For planar graphs, $Z_{\mathcal G}(n)=\Theta(n)$~\cite{BlellochF10}.
Let $G=(V, E)$ be a graph. 
It is well known that planar graphs are closed under taking minors and have only
linearly many edges, that is, $m=O(n)$.
For vertex sets $V'\subseteq V$, we
write $G[V']$ for the subgraph of $G$ induced by $V'$. For a vertex $u \in V$,
we write $G-u \coloneqq G[V \setminus \{u\}]$. A \emph{separator} is a vertex set
$S\subseteq V$ such that $G[V\setminus S]$ is disconnected. It is
\emph{balanced} if $V\setminus S$ can be partitioned into two sets $A$ and $B$
such that there is no edge between $A$ and $B$, and $|A|\leq \delta n$ and
$|B|\leq \delta n$ for some fixed constant $\delta\in(0,1)$. 

\smallpar{$r$-Divisions.}
An $r$-division $\mathcal{P}$ of a graph $G$ is a collection of subgraphs $P
\in \mathcal{P}$, called \emph{pieces}. A vertex occurrence in a piece $P$ is
called a \emph{boundary vertex} of $P$ if the corresponding vertex of $G$ is
incident to edges assigned to different pieces. All other vertex occurrences
are called \emph{non-boundary vertices}. An $r$-division has the following
properties. First, every edge of $G$ is contained in exactly one piece $P \in
\mathcal{P}$. Second, there are $\Theta(n/r)$ pieces. Third, each piece
contains at most $r$ vertices. Finally, there are $O(\sqrt{r})$ boundary
vertices per piece. For planar graphs, $r$-divisions always exist and can be
computed in linear time~\cite{Goodrich95}. Roughly speaking, an $r$-division is
constructed by a clever recursive use of balanced separators. For our purposes,
it is enough to use the relaxed definition from \cite{ElberfeldKM25}, where the
number of boundary vertices per piece can be $O(r^{1/2+\eps})$, for any fixed
$\eps\in(0,1/2)$. Elberfeld et al.~\cite{ElberfeldKM25} used this relaxation
because (1) it makes no difference for later bounds, and (2) there is an
$O(n)$-bit and $O(n)$-time algorithm that constructs an $r$-division with this relaxed
boundary size when the input graph is a planar graph~\cite{KammerM22}. For more
information on $r$-divisions,
see~\cite{Frederickson87,Goodrich95,planarity,KleinMS13}.

\smallpar{Nested Division.}
For the remainder of this section, let $G=(V,E)$ be a planar graph. We define a
\emph{nested division} as follows. Let $r_0$ and $r_1$ be two parameters with
$r_1<r_0$. First, the graph $G$ is decomposed into an $r_0$-division
$\mathcal{P}\coloneqq\{P_1,\ldots,P_N\}$ with $N=\Theta(n/r_0)$, whose pieces we
call \emph{mini pieces}. For every mini piece $P_i$, we then compute an
$r_1$-division $\mathcal{P}_i\coloneqq\{P_{i,1},\ldots,P_{i,N_i}\}$ of $P_i$,
with $N_i=\Theta(r_0/r_1)$; the pieces of these divisions are called
\emph{micro pieces}. Boundary vertices of $\mathcal P$ are called mini-boundary
vertices, and boundary vertices of the divisions $\mathcal P_i$ are called
micro-boundary vertices. Throughout the paper, we fix piece sizes
$r_0=\log^\alpha n$ and $r_1=(\log\log n)^\alpha$, where $\alpha$ is a
sufficiently large constant. Recall that each mini piece has
$O(r_0^{1/2+\eps})=o(r_0)$ mini-boundary vertices, and each micro piece has
$O(r_1^{1/2+\eps})=o(r_1)$ micro-boundary vertices, for an arbitrarily small
fixed $\eps \in (0,1/2)$. We choose $\alpha$ large enough so that $\alpha\geq
4$ and $\alpha(1/2-\eps)>2$. Then $r_0\geq c\log^4 n$ and $r_0^{1/2-\eps}\geq
c\log^2 n$, and analogously $r_1\geq c\log^4\log n$ and $r_1^{1/2-\eps}\geq
c\log^2\log n$, for any fixed constant $c$ and sufficiently large $n$.
Equivalently, the size of a mini piece is larger than its boundary by at least
a factor of $\log^2 n$, and the size of a micro piece is larger than its
boundary by at least a factor of $\log^2\log n$. 
With this choice of $\alpha$,
the number of mini pieces is $O(n/\log^4 n)$, and the total number of mini-boundary vertices is
$O(n/\log^2 n)$. Analogously, over all mini pieces, the number of micro
pieces is $O(n/\log^4 \log n)$ and the total number of micro-boundary vertices is $O(n/\log^2\log
n)$. To avoid carrying the constants through every expression, we write $\plgn$
and $\plglgn$ for sufficiently large polylogarithmic functions chosen so that
all bounds above hold. For example, we say there are $O(n/\plgn)$ mini-boundary
vertices in total.

The succinct encoding of Blelloch and Farzan~\cite{BlellochF10} can be viewed
as an encoding of a nested division. Their construction first decomposes the
input graph into what they call mini graphs, and then decomposes each mini
graph into micro graphs. Up to terminology, these are the mini and micro pieces
above. Micro pieces are small enough to be represented by indices into a lookup
table, while the remaining structures store how the different pieces fit
together. In the work of Elberfeld et al.~\cite{ElberfeldKM25}, the same
underlying representation is viewed as a data structure for nested divisions,
referred to as \emph{succinct nested division}.

\smallpar{Succinct Nested Division.}
A succinct nested division of $G$ assigns ids to all pieces on both levels.
Each mini piece has a unique id $i$, and we write $P_i$ for the mini piece with
id $i$. For every mini piece $P_i$, each micro piece in the division of $P_i$
has a unique id $j$, and we write $P_{i,j}$ for the micro piece with id $j$
inside $P_i$. The encoding also assigns labels to vertex occurrences on all
three levels. Since $G$ is unlabeled, the encoding chooses \emph{global labels}
from $[n]$ for the vertices of $G$. Inside a mini piece $P_i$, vertex
occurrences are labeled locally by labels from $[r_0]$, called \emph{mini
labels}. Thus, a vertex occurrence in $P_i$ is identified by a pair $(i,u_i)$,
where $u_i$ is its mini label in $P_i$. Inside a micro piece $P_{i,j}$, vertex
occurrences are labeled locally by labels from $[r_1]$, called \emph{micro
labels}. Thus, a vertex occurrence in $P_{i,j}$ is identified by a triple
$(i,j,u_{i,j})$, where $u_{i,j}$ is its micro label in $P_{i,j}$. Mini-boundary
vertices may have several mini labels, one in each mini piece in which they
occur. Analogously, micro-boundary vertices may have several micro labels,
distinguished by the ids of the micro pieces in which they occur. Each mini
piece $P_i$ contains at most $O(r_0^{1/2+\eps})$ mini labels of boundary
vertices. Since there are $\Theta(n/r_0)$ mini pieces, the total number of
mini-boundary labels is $O(n/r_0^{1/2-\eps})$; the analogous bound holds for
micro-boundary labels inside the mini pieces. Essentially, the number of
mini/micro labels of mini/micro-boundary vertices is asymptotically the same as
the number of mini/micro-boundary vertices. This is formalized in the next
observation.
\begin{observation}\label{obs:multiplicities}
Over all mini pieces, there are $O(n/\plgn)$ mini-boundary vertices and
$O(n/\plgn)$ mini-boundary labels. Over all micro pieces, there
are $O(n/\plglgn)$ micro-boundary vertices and $O(n/\plglgn)$ micro-boundary
labels.
\end{observation}

A succinct nested division provides access to \emph{translation mappings}
between the levels. If a vertex is identified by its global label $u$, we can
obtain all pairs $(i,u_i)$ that represent occurrences of $u$ in mini pieces
$P_i$. Conversely, if we are given the id $i$ of a mini piece and a mini label
$u_i$ inside $P_i$, then we can recover the global label $u$ of $u_i$. The same
translation is available between mini- and micro labels inside each mini piece:
from a mini label $u_i$ in $P_i$ we obtain all pairs $(j,u_{i,j})$ representing
its occurrences in micro pieces, and from $(i,j,u_{i,j})$ we can obtain the
corresponding mini label. With these translation mappings, Blelloch and
Farzan~\cite{BlellochF10} realize \emph{basic graph queries}: adjacency,
degree, and neighborhood queries. These queries are provided for the global
graph and locally inside mini and micro pieces.

We use the \emph{table-swap} operation of Elberfeld et
al.~\cite[Section~2.3]{ElberfeldKM25}. Recall that every micro piece is
represented by an index into a lookup table. This lookup table does not have to
contain only the underlying micro pieces. Instead, it can be extended so that
its entries also store a small amount of additional information. For each micro
piece we can store a constant number of orderings and a constant number of
colorings of at most $r_1^\gamma$ distinguished vertices, for any fixed $\gamma
\in (0,1)$, where the colors come from a universe $[\lceil r_1^c\rceil]$ for
any constant $c$. In addition, the table entry itself may encode a constant
number of colors from such a universe. The table then stores all possible such
orderings and colorings for every graph in the table. This idea is already
implicit in the encoding of Blelloch and Farzan~\cite{BlellochF10}, e.g.,
for encoding combinatorial embeddings; earlier variants appear
in, e.g.,~\cite{AleardiDS05}.
For us, the distinguished vertices are equal to the micro boundary
vertices.
Entries of
the lookup table may be seen as micro pieces together with a local state. For
every constant number of local queries or updates whose result can be computed
in polynomial time for one table entry, the answers can be precomputed as part
of the lookup table. Suppose that the micro piece
$P_{i,j}$ is currently stored by an index $x$. 
A table query may interpret $x$ as encoding not only the underlying micro
piece, but also the state of an algorithm.
The answer to the query is another index $y$ that
encodes the same micro piece together with the updated state. The
table-swap operation is later used to advance a BFS state inside micro pieces.

\begin{lemma}[\cite{BlellochF10,ElberfeldKM25,KammerM22}]
\label{lem:nested}
Let $G$ be a planar graph. There exists a succinct nested division $\mathcal D$
that provides basic graph queries, translation mappings, and table-swap
operations.
\end{lemma}

\section{Executing a BFS inside a Succinct Nested Division}
\label{sec:bfs}

We describe how to execute a breadth-first search directly on a succinct nested
division, building the BFS tree layer by layer. In this section, let $G=(V,E)$
be a connected planar graph given by a succinct nested division $\mathcal D$,
and let $s\in V$ be the start vertex of the BFS.

\smallpar{Active Pieces.}
For a vertex $u$, we write $\ell(u)$ for its BFS layer, and $L_d\coloneqq\{u\in
V\mid \ell(u)=d\}$ for the $d$-th BFS layer. Round $d$ of the BFS processes the
vertices of $L_d$ and attaches all previously unreached neighbors as vertices
of $L_{d+1}$. Observe that the order in which vertices of $L_{d+1}$ are
discovered is irrelevant. Instead of storing the set $L_d$ explicitly, we
maintain the pieces that still have to be processed in the current round. We
call such pieces \emph{active}. More precisely, the algorithm maintains a list
of active mini pieces $P_i$, and for each $P_i$, a list of active micro pieces
$P_{i,j}$ inside $P_i$. In addition, for each active mini piece $P_i$, we store
the mini-boundary vertices of $P_i$ that lie in the current layer and have to
be processed, using their global labels. Processing such a mini-boundary vertex
$u$ in $P_i$ means iterating over all edges incident to its occurrence $u_i$
inside $P_i$. For
every previously unreached neighbor reached through such an edge, we assign
layer $d+1$ and store $u$ as its parent. Processing an active mini piece $P_i$
consists of processing these active mini-boundary vertices explicitly and
processing all active micro pieces inside $P_i$. Processing an active micro
piece advances the BFS locally inside the micro piece by a table-swap
operation. This may also create side effects outside the micro piece when
boundary vertices are reached; these side effects and the local micro-piece
state are defined concretely in the next paragraphs. Layer
and parent information is stored as follows. For every mini-boundary vertex
$u$, we store $\ell(u)$ and its parent explicitly. For a micro-boundary vertex
$w$ given by its mini label $w_i$ inside a mini piece $P_i$, we store a
reference to the mini label $u_i$ of the closest reached mini-boundary vertex
$u$ of $P_i$ on the root-to-$w$ path in the BFS tree, together with the offset
$\ell(w)-\ell(u)$. The parent of $w$ is stored explicitly by its mini label in
$P_i$. For vertices strictly inside micro pieces, offsets, references, and
parent information are computed on the fly via table lookup.
Whenever we reach a previously unreached vertex, the vertex is
assigned to layer $d+1$. If the reached vertex lies in another mini or micro
piece, the corresponding piece is marked active for the next round. Thus,
processing the active pieces of round $d$ constructs the active pieces of
round $d+1$.

\smallpar{Partial BFS States of Micro Pieces.}
Let $P_{i,j}$ be a micro piece. The local state of $P_{i,j}$ is represented by
a \emph{partial BFS state}, consisting of a \emph{boundary profile} and a
\emph{local advance}. The boundary profile stores a constant number of colors
per micro-boundary vertex of $P_{i,j}$, together with an ordering of these
boundary vertices. The assigned colors encode whether a boundary vertex is
visited or unvisited by the BFS and additionally mark boundary vertices whose
layer is currently being processed. For every visited boundary vertex, they
also store its \emph{visitation offset}, that is, the number of local advances
already completed when the vertex is inserted into the profile. The local advance records
the total number of advances completed in $P_{i,j}$. Thus, the state records
how many local advances separate the visits of two boundary vertices and how
far the search has progressed since each was reached. These values are bounded
by the size of $P_{i,j}$ and can therefore be represented by the colors
supported by the lookup table. The ordering encodes the relative order of the
boundary layers, including ties. No absolute layer values, parents, or
references to mini-boundary labels are stored in the table state. These values
are stored externally for the micro-boundary labels, as described above. Given
the boundary profile and the local advance, the lookup table determines all
local information inside $P_{i,j}$: the reached internal vertices, their local
layer offsets, their local parent choices, and the current local layer to be
processed.
A valid next state is obtained by processing the current local layer using only
edges assigned to $P_{i,j}$. If several valid next states are possible, e.g.,
due to different parent choices, we choose one arbitrarily and fix it in the
lookup table.
If no new micro-boundary vertex is reached, only the local advance
changes. If new micro-boundary vertices are reached, then their colors and
their positions in the boundary order are updated. A table-swap query returns
the new state, the list of newly reached micro-boundary vertices, and a flag
indicating whether $P_{i,j}$ should remain active in the next round. This flag
is set exactly if, after the local update, $P_{i,j}$ contains a reached vertex
of the next BFS layer whose incident edges assigned to $P_{i,j}$ still have to
be processed.

\begin{lemma}\label{clm:table}
The next partial BFS state of an active micro piece $P_{i,j}$ can be computed
by one table-swap query. The query reports the micro-boundary vertices of
$P_{i,j}$ reached for the first time in the next state and whether $P_{i,j}$
should remain active.
\end{lemma}

\smallpar{Exceptional Pieces.}
There are two exceptional cases that we handle explicitly. First, suppose the
start vertex $s$ is not a mini-boundary vertex and lies in some mini piece
$P_i$. Then, at the beginning of the BFS, the search inside $P_i$ has no
reached mini-boundary vertex that can serve as a reference for storing relative
layer offsets, as assumed above. We therefore handle the starting mini piece
explicitly, treating all vertices of $P_i$ as if they were mini-boundary
vertices for the purpose of storing layer and parent information. This affects
only one mini piece of size $O(\plgn)$ and is thus negligible. The second
exception concerns mini labels of mini-boundary vertices that are not
represented by micro-boundary labels. Let $u$ be such a vertex, let $u_i$ be a
mini label of $u$ in a mini piece $P_i$, and let $P_{i,j}$ be the unique micro
piece of $P_i$ containing $u_i$, where $u_i$ has micro label $u_{i,j}$. When
$u$ is reached, its layer and parent information must be made available inside
$P_{i,j}$. However, the table index of $P_{i,j}$ stores the boundary profile
only for micro-boundary labels, and therefore does not allow changing the state
of $u_{i,j}$ directly. We handle every such micro piece using explicit visitation
flags, layer and parent records, and queues. There are $O(n/\plgn)$ micro pieces containing mini labels of
mini-boundary vertices, each of size $O(\plglgn)$, so their total size is
$o(n)$. Hence, all exceptional pieces use $o(n)$ bits and contribute only
$o(n)$ time. All remaining micro pieces are processed by table-swap operations. 

\smallpar{Boundary synchronization.}
Consider an active micro piece $P_{i,j}$ during round $d$. By \Cref{clm:table},
the table-swap query for $P_{i,j}$ reports exactly those micro-boundary
vertices that are reached for the first time by the local update. Let $w_i$ be
the mini label of such a vertex inside $P_i$. Then the new layer and parent
information of $w_i$ must be copied to all micro labels $w_{i,j'}$ of $w_i$
in the micro pieces $P_{i,j'}$ containing it. We call such a propagation a
\emph{boundary synchronization}; see \Cref{fig:cascade}.
We first advance all active pieces of the current round and then synchronize
the newly reached boundary vertices, setting their visitation offsets using
the updated local-advance counters.
There are $O(n/\plglgn)$ micro-boundary labels over all micro pieces
(Obs.~\ref{obs:multiplicities}). Each such label is updated at most once,
namely when the corresponding vertex is reached for the first time. Updating a
label means changing the boundary profile of its micro piece $P_{i,j}$ by one table-swap
operation and, if necessary, marking $P_{i,j}$ active for the next
round. Thus, all micro-boundary synchronizations cause $O(n/\plglgn)$
table-swap operations in total. The same argument applies to mini-boundary
vertices: when a mini-boundary vertex $u$ is reached for the first time, all
mini labels $u_i$ of $u$ are marked as reached, and the corresponding mini
pieces are marked active for the next round. Since there are only $O(n/\plgn)$
mini-boundary labels in total, these updates are negligible compared with the
micro-boundary synchronizations.

\begin{figure}[ht!]
\centering
\begin{tikzpicture}[
    piece/.style={draw, line width=0.9pt},
    leftpiece/.style={piece, fill=blue!10},
    rightpieceA/.style={piece, fill=green!12},
    rightpieceB/.style={piece, fill=green!35},
    layer/.style={draw=gray!75, dashed, line width=0.9pt},
    newlayer/.style={draw=red!75, thick},
    bvertex/.style={circle, draw, fill=black, inner sep=0pt, minimum size=4.8pt}
]
\begin{scope}[yscale=0.5]
\coordinate (LLsw) at (-6.4,-1.5);
\coordinate (LLne) at (-4.2, 1.5);

\coordinate (Lsw)  at (-4.2,-1.5);
\coordinate (Lne)  at ( 0.0, 1.5);

\coordinate (RTsw) at ( 0.0, 0.0);
\coordinate (RTne) at ( 3.2, 1.5);

\coordinate (RBsw) at ( 0.0,-1.5);
\coordinate (RBne) at ( 3.2, 0.0);

\draw[leftpiece]   (LLsw) rectangle (LLne);
\draw[piece]   (Lsw)  rectangle (Lne);
\draw[rightpieceA] (RTsw) rectangle (RTne);
\draw[rightpieceB] (RBsw) rectangle (RBne);

\begin{scope}
    \clip (Lsw) rectangle (Lne);
    \coordinate (o) at (-2,-1.5);
    \draw[layer]    (o) circle[radius=0.50];
    \draw[layer]    (o) circle[radius=1.00];
    \draw[layer]    (o) circle[radius=1.50];
    \draw[layer]    (o) circle[radius=1.90];
    \draw[newlayer] (o) circle[radius=2.50];
\end{scope}

\begin{scope}
    \clip (RTsw) rectangle (RTne);
\end{scope}

\begin{scope}
    \clip (RBsw) rectangle (RBne);
    \coordinate (o) at (3,-1.5);
    \draw[layer] (o) circle[radius=0.50];
    \draw[layer] (o) circle[radius=1.00];
\end{scope}

\begin{scope}
    \clip (LLsw) rectangle (LLne);
    \coordinate (o) at (-6,-1.5);
    \draw[layer] (o) circle[radius=0.50];
    \draw[layer] (o) circle[radius=1.00];
    \draw[layer] (o) circle[radius=1.50];
\end{scope}

\node[bvertex,label=above left:$u$] at (-4.2,-.21) {};
\node[bvertex,label=above right:$v$] at (0,0) {};
\end{scope}
\end{tikzpicture}
\caption{Illustration of two boundary synchronizations in a micro piece. The
two boundary vertices $u$ and $v$ must update their status in all other micro
pieces in which they are contained.}
\label{fig:cascade}
\end{figure}
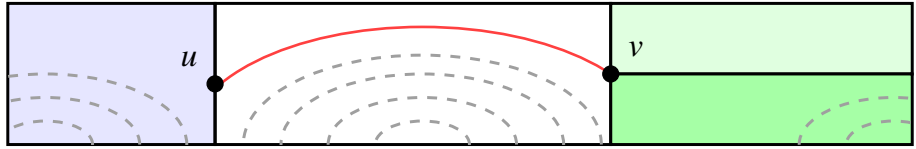

\begin{lemma}\label{clm:synchronization}
All boundary synchronizations can be processed in $o(n)$ time total.
\end{lemma}

\smallpar{Running Time, Space, and Correctness.}
Recall that the BFS proceeds in rounds: round $d$ finds all vertices of layer
$L_{d+1}$ by processing the active mini and micro pieces of layer $L_d$. We
first analyze the cost of maintaining and processing one round, and then sum
over all rounds. Active pieces are stored by lists of piece ids. To suppress
duplicate insertions, we keep bitvectors for the current and next round: one
bitvector over the mini pieces, and, inside each mini piece $P_i$, one
bitvector over the micro pieces of $P_i$. 
Between rounds, the entries set in these bitvectors are cleared by scanning the
corresponding active lists and resetting the respective entries to zero.
For each active mini piece $P_i$, we additionally store the list of active
mini-boundary vertices of $P_i$. 
The bitvectors over all mini and micro pieces use $O(n/\plglgn)=o(n)$
bits. The active lists also use $o(n)$ bits: mini-piece ids need $O(\log n)$
bits, while micro-piece ids are stored locally inside their mini pieces and
need only $O(\log\log n)$ bits, and the polylogarithmic exponents were chosen
sufficiently large.
The layer and parent information also uses $o(n)$ bits:
mini-boundary vertices store global layer and parent values explicitly,
micro-boundary vertices store offsets and mini labels, 
both using $O(\log\log n)$ bits,
and internal vertices of
micro pieces are handled by the table lookup. Each round processes all active
mini pieces and active micro pieces as follows. For every active mini piece
$P_i$, we process its active micro pieces $P_{i,j}$ by table-swap operations.
Each such operation processes the current layer inside $P_{i,j}$. Summed over
all rounds, the number of table-swap operations is $O(n)$, since each vertex
occurrence is processed only in its BFS layer. An active mini-boundary
vertex $u$ is processed explicitly in each mini piece $P_i$ containing it
by scanning the edges incident to its mini label $u_i$ inside $P_i$. 
Since each edge is assigned to
exactly one mini piece and vertices are processed only in their BFS layer,
these scans contribute $O(n)$ time in total. All
boundary synchronizations take $o(n)$ time (\Cref{clm:synchronization});
exceptional pieces contribute only $o(n)$ time and space.

Correctness follows by induction over the rounds. At the beginning of round
$d$, every vertex at distance at most $d$ from $s$ has been reached, all
vertices of layer $L_d$ that still have incident assigned edges to process are
represented by active pieces, and all reached boundary labels have been
synchronized. Processing all active mini pieces therefore considers every edge
incident to a vertex of $L_d$: edges incident to active mini-boundary vertices
are scanned explicitly, and all remaining edges are processed inside active
micro pieces. Every previously unreached neighbor reached in this way is
assigned to layer $d+1$ and receives a parent in $L_d$. Conversely, every
vertex at distance $d+1$ has a neighbor in $L_d$, and the edge to such a
neighbor is considered during round $d$, so the vertex is reached in this
round. Thus, the invariant holds for round $d+1$. When no active piece remains,
all vertices have been reached, and the stored parent pointers form a BFS tree
rooted at $s$.

\begin{lemma}\label{lem:bfs}
Let $\mathcal D$ be a succinct nested division of a connected planar graph
$G=(V,E)$. A BFS tree rooted at any $s \in V$ can be computed in $O(n)$ time
using $o(n)$ additional bits.
\end{lemma}

\smallpar{Basic Operations on the BFS Tree.}
Operations on the BFS tree are supported in the same way
as outlined in~\cite{ElberfeldKM25} for DFS trees.
While not stated as such, their construction
does not rely on DFS-specific properties. For completion's sake,
we give a concrete proof for all operations outlined 
in the next lemma in Appendix~\ref{app:treequeries}.

\begin{lemma}\label{lem:treequeries}
For any rooted spanning tree $T$ represented in a succinct nested division,
we can traverse $T$ in preorder, postorder, or breadth-first order; report the
depth and parent of a vertex; iterate over its children; test whether an edge
belongs to $T$ and answer lowest-common-ancestor queries.
\end{lemma}

Together, \Cref{lem:bfs,lem:treequeries} give \Cref{thm:main}. The
encoding can be reset to its original state in $o(n)$ time by iterating
over all boundary vertices and micro pieces.

\section{Constructing the Encoding of Theorem~\ref{thm:main}}\label{app:construction}
We give a brief outline of how the encoding of \Cref{thm:main} is constructed
and how the bounds of \Cref{cor:construction} are achieved. 
Recall that
the underlying data structure of \Cref{thm:main} is essentially the succinct
encoding of Blelloch and Farzan~\cite{BlellochF10}. Its construction has three
main ingredients: the computation of the nested divisions, the construction of
the lookup table for micro pieces, and the construction of the translation
mappings between global labels, mini labels, and micro labels. 
Additional data stored for boundary vertices is negligible in the bounds below.
First, we discuss the construction of the nested division. Recall that standard
planar $r$-division algorithms can compute divisions in linear time, but use
$O(n\log n)$ bits of working space~\cite{Goodrich95} and that we use the relaxed
definition of Elberfeld et al.~\cite{ElberfeldKM25}, where each piece has
$O(r^{1/2+\varepsilon})$ boundary vertices, for any fixed
$\varepsilon \in (0,1/2)$, instead of the standard $O(r^{1/2})$ bound. For
this relaxed version, Hammer, Kammer, and Meintrup~\cite{KammerM22} give an
$O(n)$-time and $O(n)$-bit construction for planar graphs; they also mention
the construction of the succinct encoding of Blelloch and Farzan as an
application. Recursively applying this construction gives the nested division
within the claimed construction bounds.
The lookup table contributes only $o(n)$ time and $o(n)$ space.
Since micro pieces have size at most
$r_1=\plglgn$, the table has size $o(n)$, including the additional bits
needed to support the table-swap operations we use. Since the table entries are
so small, any polynomial-time precomputation for all entries can be done in $o(n)$
time total as well.
It remains to discuss the translation mappings. These mappings are implemented
using compressed indexable dictionaries. In particular, they use the compressed
dictionary of Raman et al.~\cite{RamanRRS07}, which stores a set
$S \subseteq [\ell]$ in $\log \binom{\ell}{|S|} + o(\ell)$
bits while supporting rank and select queries in constant time. In our setting,
the universe has size $\Theta(n)$, while the stored sets correspond to
boundary vertices and piece ids and have size only
$O(n/\plgn)$ or $O(n/\plglgn)$. Hence, all required compressed
dictionaries use $o(n)$ bits in total. No deterministic linear-time
construction of these compressed dictionaries is known, but they can be
constructed in expected linear time~\cite{FeigenblatPS16,RamanRRS07}. This
yields the expected linear-time construction of the succinct encoding.
If succinctness is not required, the compressed dictionaries can be replaced by
standard indexable dictionaries, see, e.g.,~\cite{BaumannH19,Jacobson88}.
These do not exploit the imbalance between the
universe size and the stored set size; they store a bitvector of length $\ell$
together with auxiliary data structures using $\ell+o(\ell)$ bits, and can be
constructed deterministically in linear time. With this replacement, the nested
division is no longer succinct, but compact, and can be constructed in
deterministic $O(n)$ time using $O(n)$ bits.

\section{Generalization to Other Graph Classes}\label{app:general}

A subgraph-closed graph class $\mathcal G$ is called \emph{separable} if
every $n$-vertex graph in $\mathcal G$ has a balanced separator of size
$O(n^c)$ for some fixed constant $c\in(0,1)$. The encoding of Blelloch and
Farzan~\cite{BlellochF10} and its augmentation by Elberfeld et
al.~\cite{ElberfeldKM25} apply to separable graph classes. The piece-size
exponents must be chosen in accordance with the separator bound: with
$\beta=1-c$, choose $\alpha\beta>2$ and increase $\alpha$ if necessary.
Then all additional data still occupy $o(n)$ bits. The number of local
advances is bounded by the size of the piece, independently of planarity, so the same state
reconstruction and synchronization apply.

Once a succinct nested division with the interface of \Cref{lem:nested}
is available, our BFS and tree-query methods apply without a planar
embedding: they use only local graph operations, table swaps, and the
translation mappings. Separable classes have $O(n)$ edges, so processing
their edges still takes $O(n)$ total time. Thus, the BFS takes $O(n)$ time
and $o(n)$ additional bits and leaves its tree available for the same queries.
The expected linear-time, $O(n)$-bit construction and the deterministic
compact variant of \Cref{cor:construction} are guaranteed here only for
planar graphs.

\bibliography{main}

\clearpage
\appendix
\section{Proof of Lemma~\ref{lem:treequeries}}\label{app:treequeries}

We give the details of the tree queries used in \Cref{thm:main}.
The constructions are essentially the same as those
of~\cite[Sections~3.1 and~3.2]{ElberfeldKM25}. Their use of local tree
information and explicit boundary records does not depend on the tree
being a DFS tree. We therefore describe the queries for an arbitrary
rooted spanning tree $T$ of $G$ that is already represented in the nested
division. Its construction in the BFS case is given by \Cref{lem:bfs}.

For a piece $P$, write
$F_P=(V(P),E(T)\cap E(P))$ for its local forest, i.e., the subgraph of $T$
contained within $P$.
During query preprocessing, each micro-boundary occurrence
receives a color storing its parent's micro label if its parent edge belongs
to that micro piece, and a null value otherwise.
For each micro piece, the augmented table
entry determines this forest and the parent of every internal vertex.
For a boundary vertex, its authoritative record stores its parent, while
the local parent colors identify the occurrence of the parent edge.
We retain the id of the piece owning this edge with the parent record.
Every root of a local forest is a boundary vertex of its piece unless it is
the root of $T$.

All queries below use the labels of the encoding. A constant number of
translations between global, mini-, and micro labels takes constant time.
The graph indices and the state describing $T$ remain available throughout
a traversal. Additional local queries are precomputed for each table
entry, and additional boundary fields are covered by the table interface
of \Cref{sec:prelim}.

\begin{bfsclaim}\label{clm:treeparent}
The parent of a vertex in $T$ can be returned in constant time.
\end{bfsclaim}
\begin{proof}
Read the parent from the boundary record, or obtain it from the lookup
table of the unique micro piece for an internal vertex, and translate
the returned label to a global label. The root returns a null value.
\end{proof}

\begin{bfsclaim}\label{clm:treechildren}
The children of a vertex can be enumerated in constant time per output
element.
\end{bfsclaim}
\begin{proof}
The table enumerates the children whose parent edges belong to a given
micro piece. For a boundary vertex, store lists of the mini pieces, and
within them the micro pieces, in which it has at least one child, as
in~\cite[Section~3.1]{ElberfeldKM25}. Iterating only over these nonempty
lists gives constant initialization time and constant time per child.
Fix their order by piece id and micro label, and store inverse positions
with boundary occurrences so that iteration can resume after a specified
child. There is at most one list entry per boundary occurrence; the lists
and inverse positions use $O(n/\plglgn)$ bits and are
constructed in $O(n)$ time.
\end{proof}

\begin{bfsclaim}\label{clm:treepreorder}
A preorder traversal of $T$ takes $O(n)$ time and a constant number of
working words, in addition to the child-query structures.
\end{bfsclaim}
\begin{proof}
Use the parent and child queries to traverse every tree edge in both
directions, outputting each vertex when first entered. Since child
iteration can resume after the child from which we return, no recursion
stack is needed; the traversal takes $O(n)$ time and a constant number
of working words.
\end{proof}

\begin{bfsclaim}\label{clm:treepostorder}
A postorder traversal of $T$ takes $O(n)$ time and a constant number of
working words, in addition to the child-query structures.
\end{bfsclaim}
\begin{proof}
Use the traversal of \Cref{clm:treepreorder}, but output each vertex
immediately before returning to its parent, and output the root last.
The time and space bounds are unchanged.
\end{proof}

\begin{bfsclaim}\label{clm:treedepth}
The depth of a vertex in $T$ can be returned in constant time. If $T$ is
a BFS tree of $G$, this is its BFS layer in $G$.
\end{bfsclaim}
\begin{proof}
Store depths explicitly for mini-boundary vertices. For a micro-boundary
vertex $u$ that is not mini-boundary, store its nearest mini-boundary
ancestor $a$ and their distance in $T$, which is at most $r_0-1$.
Then $\operatorname{depth}(u)=\operatorname{depth}(a)+\operatorname{dist}_T(a,u)$.
For an internal vertex, the table returns its nearest micro-boundary
ancestor and the local distance to it, so at most two such steps suffice.
A traversal of $T$ and traversals of the local forests construct these
data in $O(n)$ time; the stored depths and references use
$O(n/\plglgn)$ bits.
\end{proof}

\begin{bfsclaim}\label{clm:treeedge}
Whether an edge $\{u,v\}$ belongs to $T$ can be tested in constant time.
\end{bfsclaim}
\begin{proof}
The edge belongs to $T$ if and only if $u$ is the parent of $v$ or $v$
is the parent of $u$. Apply \Cref{clm:treeparent} to both endpoints.
\end{proof}

\begin{bfsclaim}\label{clm}
The lowest common ancestor of two vertices in $T$ can be returned in
constant time after $O(n)$ time preprocessing, using
$O(n/\plglgn)$ additional bits.
\end{bfsclaim}
\begin{proof}
We describe the reduced-tree construction used for the query; the same
construction is used for DFS trees in~\cite[Section~3.2]{ElberfeldKM25}.
For a piece $P$, let
$F_P=(V(P),E(T)\cap E(P))$
be the forest formed by the tree edges assigned to $P$, rooted according to
their orientation in $T$. The root of each component is the vertex of minimum
depth in that component. Unless the component contains the root of $T$, this
vertex is a boundary vertex, since its parent edge is assigned to another
piece.

Suppose first that two query vertices belong to the same component of $F_P$.
The unique path between them in $T$ is then contained in $F_P$, and their
lowest common ancestor can be determined locally. For a micro piece, this
answer is precomputed in the lookup table. It remains to preserve the
information required when the two vertices do not belong to the same local
component.

For this purpose, mark the roots and boundary vertices of every component of
$F_P$. Delete all branches that contain no marked vertex, and suppress every
unmarked vertex that has exactly one remaining child. We call the resulting
forest the \emph{reduced forest} $H_P$. Its edges represent paths in $F_P$.
Every lowest common ancestor of two marked vertices remains in $H_P$: it is
either marked itself or has at least two remaining children and is therefore
not suppressed. Moreover, $H_P$ has $O(b+1)$ vertices if $P$ has $b$ boundary
vertices. Indeed, every component not containing the root of $T$ has a
boundary vertex as its root, and after suppressing vertices with one child,
the number of unmarked vertices is bounded by the number of marked vertices.

For an arbitrary vertex $u$ of $P$, let $a(u)$ be the nearest ancestor of $u$
that lies on one of the retained paths before these paths are suppressed. If
$a(u)$ is a vertex of $H_P$, let $h(u)=a(u)$. Otherwise, $a(u)$ lies in the
interior of an edge of $H_P$ representing a contracted path, and we let
$h(u)$ be the lower endpoint of this edge. The values $a(u)$ and $h(u)$ are
local information and are obtained by table lookup for a micro piece.

The forests $H_P$ are combined by identifying the different occurrences of
the same boundary vertex. The resulting reduced tree or forest $H$ preserves
the ancestry between the boundary vertices and all branching vertices needed
for their LCA queries. We equip $H$ with a linear-space, constant-time LCA
structure~\cite{HarelT84}; component identifiers are stored when the
intermediate structure is a forest.

We briefly describe how a query is recovered from this reduced structure.
If the query is not answered inside one local component, compute
$h(u)$ and $h(v)$ and let
$z=\operatorname{lca}_H(h(u),h(v))$.
If $z$ differs from both $h(u)$ and $h(v)$, then $z$ is the required lowest
common ancestor. If $z=h(u)\ne h(v)$, the path from $u$ meets the retained
part of the tree at $a(u)$, and the answer is $a(u)$; the case
$z=h(v)\ne h(u)$ is symmetric. Finally, if $h(u)=h(v)$, both vertices attach
to the same contracted path, and the answer is the shallower of $a(u)$ and
$a(v)$. Hence, vertices attached inside contracted paths are handled without
storing the complete tree.

We apply this construction on both levels of the nested division. Inside every
mini piece, the reduced forests of its micro pieces handle queries between
micro-boundary vertices. The reduced forests of the mini pieces are then
combined to handle queries between mini-boundary vertices. A query is first
answered locally inside a micro piece if possible; otherwise, its two local
values are passed to the reduced forest of the surrounding mini piece. Only
if the answer does not lie there are the corresponding values passed to the
global reduced tree. Thus, each query performs only a constant number of table
lookups and reduced-tree LCA queries.

The reduced forests of the micro pieces and the attachment information are
obtained by table lookup. The remaining structures are constructed by
processing one mini piece at a time. The global reduced tree has
$O(n/\plgn)$ vertices, while the reduced forests inside the mini pieces have
$O(n/\plglgn)$ vertices in total. Including their labels and the LCA
structures, they use
$O(n/\plglgn)$
bits and can be constructed in $O(n)$ time.
\end{proof}

\begin{bfsclaim}\label{clm:treebfsorder}
The vertices of $T$ can be enumerated in nondecreasing depth in $O(n)$
time using $O(n/\plglgn)$ bits.
\end{bfsclaim}
\begin{proof}
Run the BFS of \Cref{sec:bfs} on $T$ from its root, outputting each
vertex when it is first reached. Temporary search data preserve the
represented tree and its query structures within the same time and
space bounds.
\end{proof}

Together, the claims above give the queries of
\Cref{lem:treequeries}.

\end{document}